\documentclass[copyright,creativecommons]{eptcs}
\providecommand{\event}{Intersection Types and Related Systems (ITRS 2010)} 
\usepackage{breakurl}             

\usepackage{calrsfs}
\usepackage{stmaryrd}
\usepackage{mathpartir}
\usepackage{bussproofs}
\usepackage{amsmath,amssymb,amsthm}
\usepackage{bbold}
\usepackage{color}

\usepackage{tikz}
\usetikzlibrary{shapes}

\newif\ifproofs
\proofstrue
\proofsfalse

\newif\ifcomments
\commentstrue

\newif\iflong
\longtrue

\definecolor{lucared}{rgb}{0.5,0,0}
\definecolor{lucagreen}{rgb}{0,0.3,0}

\newcommand{\eod}{
\hspace*{\fill}{$\blacksquare$}
}

\newcommand{\eoe}{
\hspace*{\fill}{$\blacklozenge$}
}

\newcommand{\ccs}{\textsc{ccs}}

\newcommand{\process}{P}
\newcommand{\varprocess}{Q}
\newcommand{\type}{t}

\newcommand{\channel}{\mathsf{c}}
\newcommand{\varchannel}{\mathsf{d}}

\newcommand{\nameset}{\mathcal{N}}

\newcommand{\procset}{\mathcal{P}}
\newcommand{\sesst}{T}
\newcommand{\varsesst}{S}

\newcommand{\orth}[1]{{#1^{\perp}}{}}
\newcommand{\biorth}[1]{\orth{\orth{#1}}}

\newcommand{\tint}{\mathtt{int}}
\newcommand{\treal}{\mathtt{real}}
\newcommand{\tbool}{\mathtt{bool}}

\newcommand{\cend}{\mathsf{\color{lucagreen}end}}
\newcommand{\cwin}{\cend}

\newcommand{\cbottom}{\mathbb{\color{lucared}0}}
\newcommand{\ctop}{\mathbb{\color{lucared}1}}

\newcommand{\cin}[1]{{?#1}}
\newcommand{\cout}[1]{{!#1}}
\newcommand{\parop}{\mathbin{|}}

\newcommand{\sand}{\wedge}
\newcommand{\sor}{\vee}
\newcommand{\bigsand}{\bigwedge}
\newcommand{\bigsor}{\bigvee}
\newcommand{\choice}{\mathbin{\oblong}}

\newcommand{\pnull}{\mathbf{\color{lucared}0}}
\newcommand{\pwin}{\mathbf{\color{lucagreen}1}}

\newcommand{\win}{\checkmark}

\newcommand{\lred}[1]{\stackrel{#1}{\longrightarrow}}

\newcommand{\wlred}[1]{\stackrel{#1}{\Longrightarrow}}

\newcommand{\nlred}[1]{\longarrownot\lred{#1}}

\newcommand{\rulename}[1]{{\textsc{(#1)}}}

\newcommand{\rapprox}[1]{\rset{r}|_1}

\newcommand{\mtop}{\procset}
\newcommand{\mbottom}{\emptyset}

\newcommand{\initc}{\mathsf{init}}
\newcommand{\init}[1]{\init(#1)}

\newcommand{\co}[1]{\overline{#1}}
\newcommand{\coset}[1]{\overline{#1}}

\newcommand{\sem}[1]{\llbracket#1\rrbracket}
\newcommand{\csem}[1]{\sem{#1}}
\newcommand{\gsem}[2]{\mathcal{G}_{#1}(#2)}

\newcommand{\orthogonal}{\mathrel{\bot}}

\newcommand{\rset}[1]{\textsc{#1}}

\newcommand{\nin}{\not\in}
\newcommand{\neqsim}{\not\eqsim}

\newcommand{\subc}{\lesssim}
\newcommand{\subs}{\lesssim}
\newcommand{\nsubs}{\not\subs}
\newcommand{\suba}{\leqslant}
\newcommand{\seqa}{=}

\newcommand{\must}{\mathrel{\Downarrow}}
\newcommand{\may}{\mathrel{\downarrow}}
\newcommand{\mayconverge}{\mathclose{\downarrow}}
\newcommand{\mustconverge}{\mathclose{\Downarrow}}

\newcommand{\seqc}{\eqsim}
\newcommand{\seqs}{\eqsim}
\newcommand{\nseqs}{\neqsim}
\newcommand{\viable}[1]{\text{$#1$ viable}}

\newcommand{\eqdef}{\stackrel{\text{\tiny\rm def}}{=}}

\newcommand{\mytitle}{Session Types = Intersection Types + Union Types}

\title{\mytitle}
\author{Luca Padovani
\institute{Dipartimento di Informatica, Universit\`a di Torino\\
Corso Svizzera 185, Torino, Italy}
\email{padovani@di.unito.it}
}
\def\titlerunning{\mytitle}
\def\authorrunning{Luca Padovani}

\newtheorem{definition}{Definition}[section]
\newtheorem{lemma}{Lemma}[section]
\newtheorem{proposition}{Proposition}[section]
\newtheorem{theorem}{Theorem}[section]
\newtheorem{corollary}{Corollary}[section]

\newtheorem{example}{Example}[section]

\begin{document}
\maketitle

\begin{abstract}
  We propose a semantically grounded theory of session types which
  relies on intersection and union types.  We argue that intersection
  and union types are natural candidates for modeling branching points
  in session types and we show that the resulting theory overcomes
  some important defects of related behavioral theories. In
  particular, intersections and unions provide a native solution to
  the problem of computing joins and meets of session types. Also, the
  subtyping relation turns out to be a pre-congruence, while this is
  not always the case in related behavioral theories.
\end{abstract}

\section{Introduction}
\label{sec:intro}

Session
types~\cite{Honda1993,HondaVasconcelosKubo98,HondaYoshidaCarbone08}
are protocol descriptions that constrain the use of communication
channels in distributed systems.  In these systems, processes engage
into a conversation by first establishing a \emph{session} on some
private channel and then carrying on the conversation within the
protected scope of the session.
The session type \emph{prescribes}, for each process involved in the
session, the sequence and the type of messages the process is allowed
to send or expected to receive at each given time.
For example, the session type $\co{a}.\co a.b$ associated with some
channel $\channel$ states that a process can use $\channel$ for
sending two $a$ messages and then waiting for a $b$ message, \emph{in
  this order}.
Names $a$ and $b$ may stand for either message types, labels, method
names and so forth, depending on the process language one is
considering.

In most session type theories it is possible to specify protocols with
\emph{branching points} indicating alternative behaviors: for example,
the session type $\co{a}.\sesst \choice \co{b}.\varsesst$ usually
means that a process chooses to send either an $a$ message or a $b$
message and then behaves according to $\sesst$ or $\varsesst$
depending on the message that it has sent; dually, the session type
$a.\sesst \choice b.\varsesst$ usually means that a process waits for
either an $a$ message or a $b$ message, and then behaves according to
the respective continuation.
In these examples, as in the session type theories cited above, one is
making the implicit assumption that the process actively choosing to
follow one particular branch is the one that sends messages, while the
process passively waiting for the decision is the one that receives
messages. In practice, it is appropriate to devise two distinct
branching operators, instead of a single one $\choice$ like in the
examples above, to emphasize this fact. This is the key intuition
in~\cite{CastagnaDezaniGiachinoPadovani09,Padovani09,BarbaneraDeLiguoro10}
where session types are studied as proper terms of a simple process
algebra with action prefixes and two choice operators: the
\emph{internal choice} $\sesst \oplus \varsesst$ denotes that the
process decides which branch, $\sesst$ or $\varsesst$, to take and
behaves accordingly; the \emph{external choice} $\sesst + \varsesst$
denotes that the process offers two possible behaviors, $\sesst$ and
$\varsesst$, and leaves the decision as to which one to follow to the
process at the other end of the communication channel.

The approach advocated
in~\cite{CastagnaDezaniGiachinoPadovani09,Padovani09} recasts session
types into well-known formalisms (process algebras) by fully embracing
their behavioral nature.
This permits the definition of an elegant, semantically grounded
subtyping relation $\subs$ for session types as an adaptation of the
well-known \emph{must} pre-order for
processes~\cite{CCSWithoutTau,testing}. Nonetheless, the resulting
theory of session types suffers from a few shortcomings.
First of all, the semantics of the external choice is a bit involved
because in some contexts it is indistinguishable from that of the
internal choice: the typical example, which is also one of the pivotal
laws of the \emph{must} pre-order, is $a.\sesst + a.\varsesst \seqs
a.(\sesst \oplus \varsesst)$ (we write $\seqs$ for the equivalence
relation induced by $\subs$).
As a direct consequence of this, the subtyping relation $\subs$ fails
to be a pre-congruence. Indeed we have $a.b \subs a.b + b.c$ but $a.b
+ b.d \nsubs a.b + b.c + b.d \seqs a.b + b.(c \oplus d)$. This
poses practical problems (one has to characterize the contexts in
which subtyping is safe) as well as theoretical ones ($\subs$ is
harder to characterize axiomatically).
Finally, recent developments of session type theories have shown a
growing interest toward the definition of \emph{meet} and \emph{join}
operators over session types~\cite{Mezzina08}, which must be defined
in an \emph{ad hoc} manner since these do not always correspond to the
internal choice and the external choice.

In this paper we propose a language of session types which uses
intersection types and union types for modeling branching points. The
idea is that when some channel is typed by the intersection type
$\co{a}.\sesst \sand \co{b}.\varsesst$ this means that the channel has
both type $\co{a}.\sesst$ and also type $\co{b}.\varsesst$, namely a
process conforming to this type can choose to send an $a$ message or a
$b$ message and then use the channel as respectively prescribed by
$\sesst$ and $\varsesst$. Dually, when some channel is typed by the
union type $a.\sesst \sor b.\varsesst$ this means that the process
does not precisely know the type of the channel, which may be either
$a.\sesst$ or $b.\varsesst$. Hence it must be ready to receive both an
$a$ message and a $b$ message. It is the message received from the
channel that helps the process disambiguate the type of the
channel. If the message does not provide enough information, the
ambiguity is propagated, hence one pivotal law of our theory is
$a.\sesst \sor a.\varsesst \seqs a.(\sesst \sor \varsesst)$.

In summary, we argue that intersection and union types are natural,
type theoretic alternatives for internal and external choices,
respectively. Furthermore, they allow us to develop a decidable theory
of session types that are natively equipped with join and meet
operators, and where the semantically defined subtyping relation is a
pre-congruence.

\paragraph{Structure of the paper.} We devote
Section~\ref{sec:processes} to presenting a process algebra, so that
we can formalize processes and correct process interactions in dyadic
sessions (i.e., we consider sessions linking exactly two processes).
We introduce session types in
Section~\ref{sec:types}, where we use the formalization of processes
from the previous section for defining their semantics. The section
includes the description of an algorithm for deciding the subtyping
relation, a type system for checking whether a process conforms to a
given session type, as well as an extended example motivating the need
to compute meet and join of session types. We conclude in
Section~\ref{sec:conclusion} with a summary of the paper and a few
hints at future research directions.
For the sake of simplicity, in this paper we restrict ourselves to
finite processes and finite types. Indeed, the relationship between
branching operators and intersection and union types is independent of
the fact that processes may or may not be infinite. On the contrary,
dealing with infinite behaviors introduces some technical
difficulties, briefly touched upon in Section~\ref{sec:conclusion},
that we plan to address in a forthcoming and more comprehensive work.
\iflong
For the sake of readability, proofs and other technical details have
been postponed to sections~\ref{sec:extra_processes}
and~\ref{sec:extra_types}.
\fi


\section{Processes}
\label{sec:processes}

\begin{table}
\caption{\label{tab:processes}\strut Syntax of processes.}
\framebox[\textwidth]{
\begin{math}
\displaystyle
\begin{array}{@{}c@{\qquad}c@{}}
\begin{array}[t]{@{}rcl@{\quad}l@{}}
  \text{\textbf{Process}} \quad \process
  & ::= & \pnull & \text{(deadlock)} \\
  & | & \pwin & \text{(termination)} \\
  & | & \alpha.\process & \text{(prefix)} \\
  & | & \process \oplus \process & \text{(internal choice)} \\
  & | & \process + \process & \text{(external choice)} \\
\end{array}
&
\begin{array}[t]{@{}rcl@{\quad}l@{}}
  \text{\textbf{Action}} \quad \alpha
  & ::= & a & \text{(input)} \\
  & | & \co a & \text{(output)} \\
\end{array}
\end{array}
\end{math}
}
\end{table}

Let us fix some notation: we let $a$, $b$, $\dots$ range over some set
$\mathcal{N}$ of \emph{action names} whose meaning is left
unspecified; we let $\process$, $\varprocess$, $\dots$ range over
\emph{processes} and $\alpha$, $\beta$, $\dots$ range over
\emph{actions}. We distinguish \emph{input actions} of the form $a$
from \emph{output actions} of the form $\co{a}$; we say that
$\co\alpha$ is the \emph{co-action} of $\alpha$ where $\co{\co a} =
a$.
We consider the simple language of sequential processes whose grammar
is described in Table~\ref{tab:processes}. Syntactically speaking the
language is a minor variation of \ccs{} without
$\tau$'s~\cite{CCSWithoutTau,HennessyBook} without relabeling,
restriction, and parallel composition.
The terms $\pnull$ and $\pwin$ denote idle processes that perform no
further action. The former is deadlocked, while the latter represents
a successfully terminated interaction (since we are going to give
processes a testing semantics, we prefer denoting success by means of
a dedicated term $\pwin$ rather than a special action as in other
theories~\cite{testing}).
The term $\alpha.\process$ denotes a process that performs the action
$\alpha$ and then continues as $\process$.
The term $\process \oplus \varprocess$ denotes a process that
internally decides whether to behave as $\process$ or as
$\varprocess$.
Finally, the term $\process + \varprocess$ is the external choice of
$\process$ and $\varprocess$ and denotes a process that externally
offers two behaviors, $\process$ and $\varprocess$, and lets the
environment decide which one it should follow. As we will see shortly,
the decision of the environment is guided, as usual, by the initial
actions performed by $\process$ and $\varprocess$.
In the following we will usually omit trailing $\pwin$'s and write,
for example, $a.\co{b}$ instead of $a.\co{b}.\pwin$.
We will also write $\procset$ for the set of all processes.

\begin{table}
  \caption{\label{tab:lts.processes}\strut Operational semantics of processes (symmetric rules omitted).}
\framebox[\textwidth]{
\begin{math}
\displaystyle
\begin{array}{@{}c@{}}
  \inferrule[\rulename{r1}]{}{
    \pwin \lred{\mathstrut\win} \pwin
  }
  \qquad
  \inferrule[\rulename{r2}]{}{
    \alpha.\process \lred{\mathstrut\alpha} \process
  }
  \qquad
  \inferrule[\rulename{r3}]{}{
    \process \oplus \varprocess \lred{\mathstrut} \process
  }
  \\\\
  \inferrule[\rulename{r4}]{
    \process \lred{\mathstrut} \process'
  }{
    \process + \varprocess \lred{\mathstrut} \process' + \varprocess
  }
  \qquad
  \inferrule[\rulename{r5}]{
    \process \lred{\mathstrut\alpha} \process'
  }{
    \process + \varprocess \lred{\mathstrut\alpha} \process'
  }
  \qquad
  \inferrule[\rulename{r6}]{
    \process \lred{\mathstrut\co a} \process'
  }{
    \process + \varprocess \lred{\mathstrut} \co a.\process'
  }
\end{array}
\end{math}
}
\end{table}

The formal meaning of processes is given by a transition system,
defined in Table~\ref{tab:lts.processes} (symmetric rules have been
omitted).  The system consists of two relations, an unlabelled one
$\lred{}$ and a labelled one $\lred{\mu}$ where $\mu$ is a
\emph{label} is an element of $\nameset \cup \coset{\nameset} \cup \{
\win \}$ and $\win \nin \nameset \cup \coset{\nameset}$ is a flag
denoting successful termination.
We extend the $\co{\,\cdot\,}$ involution to labels so that
$\co\win = \win$ and to sets of labels $\rset{a}$ so that
$\co{\rset{a}} = \{ \co\mu \mid \mu \in \rset{a} \}$.
Intuitively $\lred{}$ represents \emph{internal}, \emph{invisible}
transitions of a process, while $\lred{\mu}$ represents
\emph{external}, \emph{visible} transitions of a process. We briefly
describe the meaning of the rules in the following paragraph:
rule~\rulename{r1} signals the fact that the process $\pwin$ has
terminated successfully;
rule~\rulename{r2} states that a process $\alpha.\process$ may execute
the action $\alpha$ and reduce to $\process$;
rule~\rulename{r3} (and the symmetric one) states that a process
$\process \oplus \varprocess$ internally decides to reduce to either
$\process$ or $\varprocess$;
rule~\rulename{r4} (and the symmetric one) states that internal
decisions taken in some branch of an external choice do not preempt the
other branch of the external choice. This rule is common in process
algebras distinguishing between internal and external choices, such as
\ccs{} without $\tau$'s~\cite{CCSWithoutTau} from which out process
language is inspired.
Rule~\rulename{r5} (and the symmetric one) states that an external
choice offers any action that is offered by either branch of the
choice.
Rule~\rulename{r6} and its symmetric is possibly the less familiar
one. It states that a process performing an output action may preempt
other branches of an external choice. This rule has been originally
introduced in~\cite{CastellaniHennessy98} where the message sent is
detached from its corresponding continuation, which is thus
immediately capable of interacting with the surrounding
environment. Here, as in~\cite{CastagnaDezaniGiachinoPadovani09}, we
keep the message and its continuation attached together, so as to
model an asynchronous form of communication where the order of
messages is preserved. This is practically justified in our setting as
we aim at modelling dyadic sessions.
In the following we will sometimes use the following notation: we
write $\wlred{}$ for the reflexive and transitive closure of
$\lred{}$; we let $\wlred{\mu}$ be the composition
$\wlred{}\lred{\mu}\wlred{}$; we write $\process \nlred{}$ if there is
no $\process'$ such that $\process \lred{} \process'$; we write
$\process \wlred{\mu}$ if $\process \wlred{\mu} \process'$ for some
$\process'$; let $\initc(\process) \eqdef \{ \mu \mid \process
\wlred{\mu} \}$.

The next and final step is to describe how two processes ``complete
each other'', in the sense that they interact without
errors. Informally, $\process$ and $\varprocess$ interact without
errors if, regardless of the respective internal choices, they are
always capable of synchronizing by means of complementary actions or
they have both successfully terminated. We formalize this as the
following orthogonality relation between processes:

\begin{definition}[orthogonal processes]
  Let $\lred{}$ be the smallest relation between systems
  $\process \parop \varprocess$ of two processes such that:
\[
\inferrule{
  \process \lred{\mathstrut} \process'
}{
  \process \parop \varprocess \lred{\mathstrut} \process' \parop \varprocess
}
\qquad
\inferrule{
  \varprocess \lred{\mathstrut} \varprocess'
}{
  \process \parop \varprocess \lred{\mathstrut} \process \parop \varprocess'
}
\qquad
\inferrule{
  \process \lred{\mathstrut\co\alpha} \process'
  \\
  \varprocess \lred{\mathstrut\alpha} \varprocess'
}{
  \process \parop \varprocess \lred{\mathstrut} \process' \parop \varprocess'
}
\]
and let $\wlred{}$ be the reflexive, transitive closure of $\lred{}$.
We write $\process \parop \varprocess \nlred{}$ if there are no
$\process'$ and $\varprocess'$ such that $\process \parop \varprocess
\lred{} \process' \parop \varprocess'$.
We say that $\process$ and $\varprocess$ are \emph{orthogonal},
notation $\process \orthogonal \varprocess$, if $\process \parop \varprocess
\wlred{} \process' \parop \varprocess' \nlred{}$ implies $\process'
\lred{\win}$ and $\varprocess' \lred{\win}$.
\eod
\end{definition}

As an example, consider the process $\process \eqdef \co{a}.(a +
b)$. Then $a.\co{a}$, $a.\co{b}$, $a.(\co{a} \oplus \co{b})$ are all
orthogonal to $\process$. The processes $a$ and $\process$ are
\emph{not} orthogonal because $a \parop \process \lred{} \pwin \parop
a + b \nlred{}$ and $a + b \nlred{\win}$ (both processes must be in a
successfully terminated state when they reach a stable
configuration). Also $a.(\co{a} \oplus \co{c})$ and $\process$ are not
orthogonal because $a.(\co{a} \oplus \co{c}) \parop \process \lred{}
\co{a} \oplus \co{c} \parop a + b \lred{} \co{c} \parop a + b
\nlred{}$.

Orthogonality provides us with a notion of ``test'' that we can use
for discriminating processes, in the spirit of the testing
framework~\cite{testing}. Informally, when $\process \orthogonal
\varprocess$ we can see $\varprocess$ as a test that $\process$
succeeds to pass (since orthogonality is symmetric, we can also reason
the other way around and see $\process$ as a test for
$\varprocess$). Equivalently, we can see $\varprocess$ as a context
that completes $\process$. Then, we can say that two processes are
equivalent if they pass the same tests, or if they are completed by
the same contexts. In fact, it makes sense to interpret processes as
the set of tests they pass and to define a pre-order between
processes, which we call \emph{refinement}, as the inclusion of their
corresponding interpretations.

\begin{definition}[process interpretation and refinement]
\label{def:refinement}
  Let $\sem\process \eqdef \{ \varprocess \in \procset \mid \process
  \orthogonal \varprocess \}$.
  We say that $\varprocess$ is a \emph{refinement} of $\process$,
  notation $\process \subc \varprocess$, if and only if $\sem\process
  \subseteq \sem\varprocess$.
  We write $\seqc$ for the equivalence relation induced by $\subc$,
  namely ${\seqc} = {\subc} \cap {\subc}^{-1}$.
  \eod
\end{definition}

Intuitively, $\varprocess$ is a refinement of $\process$ if any test
that $\process$ passes is also passed by $\varprocess$. Therefore, it
is safe to replace $\process$ with $\varprocess$ as any context in
which $\process$ operates correctly will continue to do so also with
$\varprocess$.
%
%
The equational theory induced by refinement is closely related to the
\emph{must} testing pre-order~\cite{testing}.  In particular, we have
$\process \oplus \varprocess \subc \process$ since
$\sem{\process\oplus\varprocess} = \sem\process \cap
\sem\varprocess$. This equivalence lets us appreciate the fact that
the internal choice operator does correspond to an intersection when
we interpret processes as the sets of their orthogonals. Alas, under
this interpretation the external choice operator does \emph{not}
correspond to a union, for three reasons:
\begin{itemize}
\item There can be processes in $\sem{\process + \varprocess}$ that
  are not contained in $\sem\process \cup \sem\varprocess$. For
  example, $\co{a} \oplus \co{b} \in \sem{a + b} \setminus \sem{a}
  \cup \sem{b}$. This is fairly common in every framework that
  accounts for non-deterministic entities. In our case, $\co{a} \oplus
  \co{b}$ is orthogonal to $a + b$, but not to $a$ or $b$ alone.

\item Sometimes $\sem{\process + \varprocess} = \sem{\process \oplus
    \varprocess} = \sem\process \cap \sem\varprocess$, namely the
  external choice can be internal choice in disguise. For example, we
  have $a.\co{a} + a.\co{b} \seqc a.\co{a} \oplus a.\co{b} \seqc
  a.(\co{a} \oplus \co{b})$.
  The problem is that both branches of the external choice are guarded
  by the same action $a$, and since it is the \emph{initial} performed
  action to determine the chosen branch the process $a.\co{a} +
  a.\co{b}$ does not offer an external choice, but is actually
  performing an internal one.
  A different instance of this phenomenon occurs when both branches of
  an external choice are guarded by output actions, because of
  rule~\rulename{r6}. For example, we have $\co{a} + \co{b} \seqc
  \co{a} \oplus \co{b}$.

\item The fact that output actions can preempt branches of external
  choices can make such branches useless. For example $\co{a} + b
  \seqc \co{a} + \pwin \seqc \co{a}$, since $\co{a} + \process \lred{}
  \co{a}$ by rule~\rulename{r6}.
\end{itemize}

A direct consequence of these subtleties related with the external
choice is that refinement fails to be a pre-congruence. In particular,
we are now able to justify the (in)equivalences $a.b + b.d \nsubs
a.b + b.c + b.d \seqs a.b + b.(c \oplus d)$ that we have anticipated
in the introduction.


Observe that there are pathological processes that are intrinsically
flawed and cannot interact correctly with any other process.  For
example, $a \oplus b$ has no orthogonals since it is not possible to
know which message, $a$ or $b$, it is ready to receive. As another
example the process $\process = a \oplus \co{b}$ has no orthogonals:
no process interacting with it can send an $a$ message, since
$\process \lred{} \co{b}$; at the same time, a process waiting for the
$b$ message from $\process$ may starve forever since $\process \lred{}
a$.



\section{Session Types}
\label{sec:types}

In this section we introduce our language of session types, we study
their semantics, and we provide a subtyping algorithm and a type
system for checking processes against session types.

\subsection{Syntax}

\begin{table}
\caption{\label{tab:types}\strut Syntax of session types.}
\framebox[\textwidth]{
\begin{math}
\displaystyle
\begin{array}{@{}l@{\quad}l@{\quad}l@{}}
\begin{array}[t]{@{}rcl@{\quad}l@{}}
  \text{\textbf{Session type}} \quad \sesst
  & ::= & \cbottom & \text{(bottom)} \\
  & | & \ctop & \text{(top)} \\
  & | & \cwin & \text{(termination)} \\
  & | & \alpha.\sesst & \text{(prefix)} \\
  & | & \sesst \sand \sesst & \text{(intersection)} \\
  & | & \sesst \sor \sesst & \text{(union)} \\
\end{array}
\end{array}
\end{math}
}
\end{table}

We let $\sesst$, $\varsesst$, $\dots$ range over \emph{session types},
which are defined by the grammar in Table~\ref{tab:types}.
The types $\cbottom$ and $\ctop$ characterize channels which cannot be
successfully used for any interaction. We postpone a more detailed
discussion about $\cbottom$ and $\ctop$ when we will formally define
their semantics. For the time being, it suffices to say that
$\cbottom$ and $\ctop$ represent the largest and smallest element in
the lattice of session types we are about to define.
The type $\cwin$ denotes channels on which no further action is
possible. There is a fundamental distinction between $\cwin$ and the
two types $\cbottom$ and $\ctop$: $\cwin$ denotes a successfully
terminated interaction, while $\cbottom$ and $\ctop$ denote the
impossibility to carry on any interaction;
the type $\alpha.\sesst$ denotes channels on which it is possible to
perform an action $\alpha$. Actions are the same ones that occur
within processes, but the point of view is slightly different: a
process \emph{executes} an action, while a session type indicates the
possibility or the obligation for a process to execute an action.  We
will appreciate more concretely this difference in
Section~\ref{sec:checker}, where we will see that the same process can
be successfully checked against different session types.
The type $\sesst \sand \varsesst$ denotes channels that have
\emph{both} types $\sesst$ and $\varsesst$. For example $\co{a}.\cwin
\sand \co{b}.\cwin$ denotes a channel that has both type
$\co{a}.\cwin$ and also type $\co{b}.\cwin$, namely it can be used for
sending both messages $a$ and $b$.
Finally, the type $\sesst \sor \varsesst$ denotes channels that either
have type $\sesst$ or $\varsesst$. For instance the type $a.\cwin \sor
b.\cwin$ associated with a channel means that a process using that
channel must be ready to receive both a message $a$ and a message $b$,
since it does not know whether the type of the channel is $a.\cwin$ or
$b.\cwin$.\footnote{We are making the implicit assumption that ``using
  a channel'' means either sending a message on it or waiting a
  message from it and that no \emph{type-case} construct is available
  for querying the actual type of a channel.}
%
To avoid clutter, in the following we will omit trailing $\cwin$'s and
write, for instance, $\co{a} \sand \co{b}$ instead of $\co{a}.\cwin
\sand \co{b}.\cwin$ when this generates no ambiguity with the syntax
of processes.

Before giving a formal semantics to session types let us discuss a few
examples to highlight similarities and differences between them and
processes.
It should be pretty obvious that $\oplus$ and $\sand$ play similar
roles: the ability for a process $\process \oplus \varprocess$ to
autonomously decide which behavior, $\process$ or $\varprocess$, to
perform indicates that the session type associated with the channel it
is using allows both alternatives, it has \emph{both} types.
No such correspondence exists between $+$ and $\sor$. For instance,
consider $\process = a.b.\co{a} + a.c.\co{b}$ and $\sesst = a.b.\co{a}
\sor a.c.\co{b}$.
The external choice in $\process$ is guarded by the same action $a$,
meaning that after performing action $a$ the process may reduce to
either $b.\co{a}$ or to $c.\co{b}$, the choice being
nondeterministic. As we have already remarked at the end of
Section~\ref{sec:processes}, one can show that $\process$ is
equivalent to $a.(b.\co{a} \oplus c.\co{b})$, where the
nondeterministic choice between the two residual branches is explicit.
The session type $\sesst$, on the other hand, tells us something
different: we do not know whether the channel we are using has type
$a.b.\co{a}$ or $a.c.\co{b}$ and receiving message $a$ from it does
not help to solve this ambiguity. Therefore, after the message $a$ has
been received, we are left with a channel whose associated session
type is $b.\co{a} \sor c.\co{b}$. At this stage, depending on the
message, $b$ or $c$, that is received, we are able to distinguish the
type of the channel, and to send the appropriate message (either $a$
or $b$) before terminating.
In summary, $\process$ and $\sesst$ specify quite different behaviors,
and in fact while $\sesst$ is perfectly reasonable, in the sense that
there are processes that conform to $\sesst$ and that can correctly
interact with corresponding orthogonal processes, the reader may
easily verify that $\process$ has no orthogonals.

\subsection{Semantics}

Intuitively we want to define the semantics $\csem\sesst$ of a session
type $\sesst$ as a set of processes, so that session types can be
related by comparing the corresponding interpretations pretty much as
we did for processes (Definition~\ref{def:refinement}).
To assist the reader with this intuition, consider the scenario
depicted below
\begin{center}
\begin{tikzpicture}
  \node at (0,0) {$\sesst \vdash \process$};
  \node at (4,0) [cylinder, draw, minimum height=6cm, shape aspect=.5] {$\channel$};
  \node at (8.5,0) {$\varprocess \in \sem\sesst$};
\end{tikzpicture}
\end{center}
where the notation $\sesst \vdash \process$ means that $\process$,
which we will think of as the ``server'', is using the end point of
channel $\channel$ according to the session type $\sesst$. We write
$\sesst \vdash \process$ instead of $\channel : \sesst \vdash
\process$ since we assume that $\process$ acts on one channel only.
The idea is that the interpretation of $\sesst$ is the set of
``client'' processes $\varprocess$ that can interact correctly with
$\process$ when placed at the other end point of the channel
$\channel$.

Before we address the formal definition of $\csem\sesst$ we must
realize that not every set of processes makes sense when interpreted
in this way:
\begin{itemize}
\item if a server is able to interact correctly with all of the
  clients in the set $X = \{ \co a, \co b \}$, then it is also able to
  interact correctly with $\co a \oplus \co b$;

\item no server is able to interact correctly with \emph{all} of the
  clients in the set $Y = \{ \co a, b \}$ because this server would
  have to perform both an input on $a$ and an output on $b$ at the
  same time.
\end{itemize}

We conclude that neither $X$ nor $Y$ above are \emph{closed} sets of
processes that can serve as proper denotations of a session type: $X$
and $X \cup \{ \co{a} \oplus \co{b} \}$ are indistinguishable because
every server $\process$ that includes $X$ in its interpretation
includes also $X \cup \{ \co{a} \oplus \co{b} \}$; $Y$ and $\procset$
are indistinguishable because there is no server that includes $Y$ in
its interpretation just as there is no server that includes the whole
$\procset$ in its interpretation.  We therefore need a closure
operation over sets of processes, which we define in terms of
\emph{orthogonal sets}, defined as follows:

\begin{definition}[orthogonal set]
\label{def:orth}
  Let $X \subseteq \procset$. Then $\orth{X} \eqdef \{ \process \in
  \procset \mid X \subseteq \sem\process \}$.
  \eod
\end{definition}

Intuitively, the orthogonal of some set of processes $X$ is the set of
those processes that include $X$ in their interpretation.
If we go back to the problematic sets of processes described earlier,
we have $\orth{X} = \{ a + b, a + b + c, a + b + c + d, \dots\}$ and
$\orth{Y} = \emptyset$.
Clearly the orthogonal of a set $X$ flips the perspective, in the
sense that if $X$ is a set of ``clients'', then $\orth{X}$ is the set
of ``servers'' of those clients. Therefore, we define the closure as
the \emph{bi-orthogonal} $\biorth{(\cdot)}$. For instance we have
$\biorth{X} = \{ \co a, \co b, \co a \oplus \co b,\dots \}$ and
$\biorth{Y} = \procset$.
We say that a set $X$ of processes is \emph{closed} if it is equal to
its closure, namely if $X = \biorth{X}$.
The fact that $\biorth{(\cdot)}$ is indeed a closure operator is
formalized by the following result:

\begin{proposition}
\label{prop:closure}
The bi-orthogonal is a closure, namely it is extensive, monotonic, and
idempotent:
\begin{enumerate}
\item $X \subseteq \biorth{X}$;

\item $X \subseteq Y$ implies $\biorth{X} \subseteq \biorth{Y}$;

\item $\biorth{X} = \biorth{\biorth{X}}$.
\end{enumerate}
\end{proposition}
\begin{proof}
  Observe that $\orth{X} = \{ \process \in \procset \mid \forall
  \varprocess \in X : \process \orthogonal \varprocess \}$.
  Then $(\orth{(\cdot)}, \orth{(\cdot)})$ is a Galois connection (more
  precisely, a polarity) between the posets $\langle 2^\procset,
  \subseteq \rangle$ and $\langle 2^\procset, \supseteq\rangle$.
  Then it is a known fact that $\biorth{(\cdot)} = \orth{(\cdot)}
  \circ \orth{(\cdot)}$ is a closure operator on the poset $\langle
  2^\procset, \subseteq \rangle$.
\end{proof}

Then we define the interpretation of session types in terms of
closures of sets of processes, where we interpret $\sand$ and $\sor$
as set-theoretic intersections and unions.

\begin{definition}[session type semantics]
\label{def:sesst}
The semantics of a session type is inductively defined by the
following equations:
\[
\begin{array}[b]{rcl}
  \csem\cbottom & = & \mbottom \\
  \csem\ctop & = & \mtop \\
  \csem\cwin & = & \biorth{\{\pwin\}} \\
  \csem{\alpha.\sesst} & = & \biorth{\{ \co\alpha.\process \mid \process \in \csem\sesst \}} \\
  \csem{\sesst_1 \sand \sesst_2} & = & \csem{\sesst_1} \cap \csem{\sesst_2} \\
  \csem{\sesst_1 \sor \sesst_2} & = & \biorth{(\csem{\sesst_1} \cup \csem{\sesst_2})}
\end{array}
\]
\end{definition}

As we comment on the definition of $\csem\cdot{}$, it is useful to
think of $\csem\sesst$ as of the set of clients that a server using a
channel with type $\sesst$ must be able to satisfy.
Since $\cbottom$ denotes the empty set of clients, a channel typed by
$\cbottom$ is the easiest to use for a server, for the server is not
required to satisfy any process.
Dually, a channel typed by $\ctop$ is the hardest to use, for the
server is required to satify any process. As this is impossible to
achieve (there is no process that is dual of every process in
$\procset$), no server can effectively use a channel typed by $\ctop$.
From a type-theoretic point of view, $\cbottom$ and $\ctop$ represent
two dual notions of emptyness: $\cbottom$ means absence of clients,
$\ctop$ means absence of servers. Later on we will see that any
session type different from $\cbottom$ and $\ctop$ is
\emph{inhabited}, in the sense that it admits at least one client and
at least one server.
A channel typed by $\cwin$ represents those clients that are satisfied
even if they do not receive any further message. The process $\pwin$
clearly is a client of $\cwin$, but it's not the only one: any process
that guarantees the $\win$ action is a client of $\cwin$. Hence we
have $\csem\cwin{} = \{ \pwin, \pwin + a, \pwin + a + b, \dots \}$. In
particular, no process that is immediately able to emit an output is
included in this set.
Regarding the session type $\alpha.\sesst$, its clients are all those
processes that perform the co-action $\co\alpha$ and whose
continuation after $\alpha$ is in $\csem\sesst{}$. If $\alpha$ is some
input action $a$ then any process in $\csem{\alpha.\sesst}{}$ sends
$\co{a}$ (and only $\co{a}$), whereas if $\alpha$ is some output
action $\co a$ then any process in $\csem{\alpha.\sesst}{}$ guarantees
the input action $a$. For example we have $a \in
\csem{\co{a}.\cend}{}$ and $a + b \in \csem{\co{a}.\cend}{}$ but
$\co{a} \oplus \co{b} \nin \csem{a.\cend}{}$. Therefore, a server
using a channel typed by $\alpha.\sesst$ is required to provide action
$\alpha$ and to continue the interaction as specified by $\sesst$.
The intersection type $\sesst_1 \sand \sesst_2$ denotes those channels
that have both type $\sesst_1$ and type $\sesst_2$. Therefore the
servers using these channels have the freedom to use them according to
either $\sesst_1$ or $\sesst_2$. That is why the clients of $\sesst_1
\sand \sesst_2$ must be clients of both $\sesst_1$ and $\sesst_2$.
The union type $\sesst_1 \sor \sesst_2$ can be explained in a dual way
with respect to the intersection. In this case, the server is unsure
whether the channel has type $\sesst_1$ or $\sesst_2$ and consequently
it must be able to satisfy (at least) all the clients of $\sesst_1$
and all the clients of $\sesst_2$ as well.
Overall we see that intersections and unions of session types match in
a quite natural way their set-theoretic interpretation. However, note
that $\csem{\sesst_1 \sand \sesst_2} = \csem{\sesst_1} \cap
\csem{\sesst_2}$ whereas in general we have $\csem{\sesst_1 \sor
  \sesst_2} \supseteq \csem{\sesst_1} \cup \csem{\sesst_2}$.  For
example, $\co{a} \oplus \co{b} \in \csem{a.\cwin \sor b.\cwin}
\setminus (\csem{a.\cwin} \cup \csem{b.\cwin})$. There is no need to
use the closure operator on $\csem{\sesst_1} \cap \csem{\sesst_2}$
since it can be shown that this set is already closed.

We use $\csem\cdot{}$ for comparing session types. In particular we
say that $\sesst$ is a \emph{subtype} of $\varsesst$ when $\sesst$'s
clients are included in $\varsesst$'s clients:

\begin{definition}[subtype]
\label{def:subsession}
We say that $\sesst_1$ is a \emph{subtype} of $\sesst_2$, written
$\sesst_1 \subs \sesst_2$, if $\csem{\sesst_1} \subseteq
\csem{\sesst_2}$.
We write $\seqs$ for the equivalence relation induced by $\subs$,
namely ${\seqs} = {\subs} \cap {\subs}^{-1}$.
\eod
\end{definition}

Unlike the refinement relation, subtyping turns out to be a
pre-congruence with respect to all the operators of the session type
language.

\begin{proposition}
  $\subs$ is a pre-congruence.
\end{proposition}
\begin{proof}
  Immediate from the definition of $\subs$ and
  Proposition~\ref{prop:closure}(2).
\end{proof}

Equally trivial is the fact that $\sand$ and $\sor$ provide us with a
native way of respectively computing the greatest lower bound and the
least upper bound of two session types. As regards $\sand$, this is
obvious since $\csem{\sesst_1\sand\sesst_2} =
\csem{\sesst_1}\cap\csem{\sesst_2}$ by definition. For $\sor$, it
suffices to observe that $\sesst_1 \subs \varsesst$ and $\sesst_2
\subs \varsesst$ implies $\csem{\sesst_1} \cup \csem{\sesst_2}
\subseteq \csem{\varsesst}$. Since $\biorth{(\csem{\sesst_1} \cup
  \csem{\sesst_2})}$ is the smallest closed set that includes
$\csem{\sesst_1} \cup \csem{\sesst_2}$ and since $\csem{\varsesst}$ is
closed, we conclude $\csem{\sesst_1 \sor \sesst_2} =
\biorth{(\csem{\sesst_1} \cup \csem{\sesst_2})} \subseteq
\csem{\varsesst}$, namely $\sesst_1 \sor \sesst_2 \subs \varsesst$.
The following extended example shows the need to compute meets and
joins of session types in some contexts. The availability of native
unions and intersections within the language of session types makes
this task trivial.

\begin{example}[global type projection]
\label{ex:global_type}
  \emph{Global types}~\cite{HondaYoshidaCarbone08,BravettiZavattaro09}
  are abstract descriptions of interactions between two or more
  participants from a neutral point of view. For example, the global
  type
\[
  \mathtt{A} \lred{a} \mathtt{B}; \mathtt{A} \lred{b} \mathtt{B}
  \choice
  \mathtt{A} \lred{a} \mathtt{B}; \mathtt{A} \lred{c} \mathtt{B}
\]
specifies a system with two participants, here indicated by the tags
$\mathtt{A}$ and $\mathtt{B}$, which interact by exchanging messages
`$a$', `$b$', and `$c$'. In a global type, an action such as
$\mathtt{A} \lred{a} \mathtt{B}$ indicates that $\mathtt{A}$ sends an
`$a$' message to $\mathtt{B}$. Actions can be composed in sequences
(with $;$) and in alternative paths (with $\choice$). Overall, the
global type describes which sequences of interactions are possible,
but not who is responsible for which choices (hence the use of a
single operator $\choice$ in branching points).
The implementation of a global type begins by projecting it on each
participant, so as to synthesize the session type that each
participant must conform to. In this example we obtain the following
projections: the projection on $\mathtt{A}$ is $\co{a}.\co{b}$ on the
l.h.s. and $\co{a}.\co{c}$ on the r.h.s.; the projection on
$\mathtt{B}$ is $a.b$ on the l.h.s. and $a.c$ on the r.h.s. Since
$\mathtt{A}$ is the only sender, it is natural that its overall
projection is $\co{a}.\co{b} \sand \co{a}.\co{c} \seqs \co{a}.(\co{b}
\sand \co{c})$.
Since $\mathtt{B}$ is the only receiver, it must be prepared to
receive the messages from $\mathtt{A}$ regardless of which messages
$\mathtt{A}$ decides to send. Therefore, the correct projection of the
global type on $\mathtt{B}$ is $a.b \sor a.c \seqs a.(b \sor c)$,
which is the least upper bound of the projections on $\mathtt{B}$ of
the two branches.
In a language of session types with behavioral choices, this upper
bound must be computed by an \emph{ad hoc} operator, since $a.b + a.c$
would be equivalent to $a.(b \oplus c)$ which does not correspond to
the correct projection for $\mathtt{B}$.
\eoe
\end{example}

As we have anticipated, for a session type to make sense, its
interpretation must be different from both $\mbottom$ and $\mtop$.
This condition roughly corresponds to non-emptyness: a standard
``value'' type is inhabited if there exists one value of that type; a
session type is inhabited if it has at least one server and at least
one client. This explains why there are two distinct ``empty'' session
types.

\begin{definition}[viable session type]
\label{def:type_viability}
We say that the session type $\sesst$ is \emph{viable} if $\sesst
\nseqs \cbottom, \ctop$.
\eod
\end{definition}

Viability is a necessary and sufficient condition for $\sesst$ to be
implementable: if $\sesst \nseqs \cbottom$ take any $\process \in
\csem{\sesst}{}$. From the hypothesis $\sesst \nseqs \ctop$ and the
fact that $\csem{\sesst}$ is closed we also know that
$\orth{\csem{\sesst}{}} \ne \mbottom$, because $\orth{\csem\sesst{}} =
\mbottom$ implies $\biorth{\csem\sesst{}} = \mtop$. Hence there exists
$\varprocess \in \orth{\csem{\sesst}{}}$. By definition of orthogonal
set we conclude $\process \orthogonal \varprocess$.
This discussion about viability emphasizes the importance of the
orthogonal operation since the sets $\csem\sesst$ and
$\orth{\csem\sesst}$ contain precisely those processes that interact
correctly via a channel typed by $\sesst$. We conclude this section by
showing that the orthogonal operator over sets of processes
corresponds to a syntactic duality operation over session types.

\begin{theorem}[dual session type]
\label{thm:dual}
  The \emph{dual} of a session type $\sesst$ is the session type
  $\co\sesst$ obtained from $\sesst$ by turning every $\cbottom$ into
  $\ctop$, every $\ctop$ into $\cbottom$, every action $\alpha$ into
  the corresponding co-action $\co\alpha$, every $\sand$ into $\sor$,
  and every $\sor$ into $\sand$. Inductively:
\[
\begin{array}{rcl}
  \co\cbottom & = & \ctop \\
  \co\ctop & = & \cbottom \\
  \co\cwin & = & \cwin \\
  \co{\alpha.\sesst} & = & \co\alpha.\co{\sesst} \\
  \co{\sesst_1 \sand \sesst_2} & = & \co\sesst_1 \sor \co\sesst_2 \\
  \co{\sesst_1 \sor \sesst_2} & = & \co\sesst_1 \sand \co\sesst_2\\
\end{array}
\]
Then $\csem{\co\sesst} = \orth{\csem\sesst}$.
\end{theorem}

\subsection{Subtyping Algorithm}

\newcommand{\optional}[1]{\{{}#1\}}

In this section we define an algorithm for deciding the subtyping
relation. Since the interpretation of a session type is usually an
infinite set of processes, we cannot hope to derive a brute force
algorithm that is based directly on Definition~\ref{def:subsession}.
Fortunately, session types admit a particularly simple and intuitive
normal form. Therefore, we split the decision algorithm in two parts:
first we provide an effective procedure for rewriting every session
type into an equivalent normal form, which happens to be unique up to
commutativity and associativity of intersections and unions. Then, we
provide a syntax-directed algorithm that decides the subtyping
relation between session types in normal form.
In what follows we will use $n$-ary intersections and unions of the
form $\bigsand_{i\in\{1,\dots,n\}} \sesst_i$ and
$\bigsor_{i\in\{1,\dots,n\}} \sesst_i$ in place of
$\sesst_1\sand\cdots\sand\sesst_n$ and
$\sesst_1\sor\cdots\sor\sesst_n$, respectively; as usual, we let
$\bigsand_{i\in\emptyset} \sesst_i = \ctop$ and
$\bigsor_{i\in\emptyset} \sesst_i = \cbottom$ by definition.
We will also write $\sesst\optional{\sand\varsesst}_\phi$ to indicate
that the ${}\sand\varsesst$ part is present only when $\phi$ holds;
similarly for $\sesst\optional{\sor\varsesst}_\phi$.

\begin{definition}[normal form]
  We say that a session type $\sesst$ is in \emph{normal form} if
  either
\[
\sesst \equiv \bigsand_{a\in\rset{a}}
\co{a}.\sesst_a\optional{\sand \cend}_{\win\in\rset{a}}
\text{\qquad or\qquad}
\sesst \equiv \bigsor_{a\in\rset{a}} a.\sesst_a\optional{\sor
    \cend}_{\win\in\rset{a}}
\]
and $\sesst_a$ is viable and in normal form for every $a \in
\rset{a}$.
\end{definition}

A process using a channel whose associated session type is
$\bigsand_{a\in\rset{a}} \co{a}.\sesst_a\optional{\sand
  \cend}_{\win\in\rset{a}}$ may send any message $a\in\rset{a}$ and it
may decide to terminate if $\win\in\rset{a}$. After sending a message
$a$, the process must continue using the channel as specified by
$\sesst_a$.
In a dual fashion, a process using a channel whose associated session
type is $\bigsor_{a\in\rset{a}} a.\sesst_a\optional{\sor
  \cend}_{\win\in\rset{a}}$ must be ready to receive any message
$a\in\rset{a}$ and it must also be ready to terminate immediately if
no such message is received and $\win\in\rset{a}$. In case a message
$a$ is received, the process must continue using the channel as
specified by $\sesst_a$.

\begin{table}
\caption{\label{tab:simpl}\strut Simplification laws (symmetric and dual laws omitted).}
\framebox[\textwidth]{
\begin{math}
\displaystyle
\begin{array}{@{\qquad}c@{\qquad}}
\inferrule[\rulename{e-prefix}]{}{
  \alpha.\cbottom \seqa \cbottom
}
\qquad
\inferrule[\rulename{e-bottom}]{}{
  \cbottom \sand \sesst \seqa \cbottom
}
\qquad
\inferrule[\rulename{e-top}]{}{
  \ctop \sand \sesst \seqa \sesst
}
\qquad
\inferrule[\rulename{e-dist}]{}{
  \alpha.\sesst \sand \alpha.\varsesst \seqa \alpha.(\sesst \sand \varsesst)
}
\\\\
\inferrule[\rulename{e-input-end}]{
  \viable{\sesst_a}{}^{(a\in\rset{a})}
}{
  \big(\bigsor_{a\in\rset{a}} a.\sesst_a\big) \sand \cend \seqa \cbottom
}
\qquad
\inferrule[\rulename{e-input-output}]{
  \viable{\sesst_a}{}^{(a\in\rset{a})}
  \\
  \viable{\varsesst}
}{
  \big(\bigsor_{a\in\rset{a}} a.\sesst_a\big) \sand \co{b}.\varsesst
  \seqa
  \cbottom
}
\qquad
\inferrule[\rulename{e-input-output-end}]{
  \viable{\sesst_a}{}^{(a\in\rset{a})}
  \\
  \viable{\varsesst}
}{
  \big(\bigsor_{a\in\rset{a}} a.\sesst_a\sor\cend\big)
  \sand
  \co{b}.\varsesst
  \seqa
  \co{b}.\varsesst\sand\cend
}
\\\\
\inferrule[\rulename{e-input-input}]{}{
  \big(\bigsor_{a\in\rset{a}} a.\sesst_a\optional{\sor\cend}_{\win\in\rset{a}}\big)
  \sand
  \big(\bigsor_{b\in\rset{b}} b.\varsesst_b\optional{\sor\cend}_{\win\in\rset{b}}\big)
  \seqa
  \bigsor_{a\in\rset{a}\cap\rset{b}} a.(\sesst_a \sand \varsesst_a)
  \optional{\sor\cend}_{\win\in\rset{a}\cap\rset{b}} \\
}
\end{array}
\end{math}
}
\end{table}

The simplicity of normal forms is due to the fact that some behaviors
(like sending a message and receiving a message) are incompatible, in
the sense that their combination (intersection or union) yields
non-viable session types. Table~\ref{tab:simpl} presents a set of laws
that are used (from left to right) as basic simplification steps in
the computation of the normal form (symmetric and dual laws are
omitted).
Laws~\rulename{e-prefix}, \rulename{e-bottom}, and~\rulename{e-top}
state that non-viable types absorb prefixes and that $\cbottom$ and
$\ctop$ are respectively neutral for $\sor$ and $\sand$, as expected.
Law~\rulename{e-dist} shows that common actions can be factored while
preserving the combining operator. In particular, the dual law
$\alpha.\sesst \sor \alpha.\varsesst \seqs \alpha.(\sesst \sor
\varsesst)$ distinguishes subtyping from refinement and from the
\emph{must} pre-order, where the law $\alpha.\process +
\alpha.\varprocess \seqc \alpha.(\process \oplus \varprocess)$ holds.
Rules~\rulename{e-input-end} and~\rulename{e-input-output} show that
no client that sends a message $a\in\rset{a}$ can be satisfied by a
server that may decide to terminate the interaction or to send a
message. This is because the action of sending a message is
irrevocable (see rule~\rulename{r6} in the transition system of
processes).
Rule~\rulename{e-input-output-end} shows that among the clients that
either send a message $a\in\rset{a}$ or terminate are those that can
also receive message $b$.
Finally, rule~\rulename{e-input-input} shows that the clients of a
server will send only messages that can surely be received by the
server. For example, $(a \sor b \sor c) \sand (b \sor c \sor d) \seqs
b \sor c$. The dual law concerns messages that can be sent by the
server. Thus $(\co{a} \sand \co{b} \sand \co{c}) \sor (\co{b} \sand
\co{c} \sand \co{d}) \seqs \co{b} \sand \co{c}$: if the server is
unsure whether the type of the channel is $\co{a} \sand \co{b} \sand
\co{c}$ or $\co{b} \sand \co{c} \sand \co{d}$, then it can only send
those messages that can travel along the channel in both cases.

\begin{lemma}
\label{lem:simplify}
The laws in Table~\ref{tab:simpl} are sound.
\end{lemma}

The simplification laws, and the axiomatization of $\subs$ that we are
about to present, would be simpler if one could prove that $\sand$ and
$\sor$ distribute over each other. We conjecture that the lattice of
closed sets of processes ordered by set inclusion is indeed
distributive (in the process language, the internal and external
choices distribute over each other), but the proof appears to be
non-trivial.

\begin{lemma}[normal form]
\label{lem:nf}
For every session type $\sesst$ there exists $\varsesst$ in normal
form such that $\sesst \seqs \varsesst$.
\end{lemma}

\begin{table}
\caption{\label{tab:axioms}\strut Subtyping algorithm.}
\framebox[\textwidth]{
\begin{math}
\displaystyle
\begin{array}{@{\qquad}c@{\qquad}}
\inferrule[\rulename{s-bottom}]{}{
  \cbottom
  \suba
  \bigsand_{a\in\rset{a}} \co{a}.\sesst_a\optional{\sand\cend}_{\win\in\rset{a}}
}
\qquad
\inferrule[\rulename{s-top}]{}{
  \bigsor_{a\in\rset{a}} a.\sesst_a\optional{\sor\cend}_{\win\in\rset{a}}
  \suba
  \ctop
}
\qquad
\inferrule[\rulename{s-end}]{}{
  \bigsand_{a\in\rset{a}} \co{a}.\sesst_a\sand\cend
  \suba
  \bigsor_{b\in\rset{b}} b.\varsesst_b\sor\cend
}
\\\\
  \inferrule[\rulename{s-input}]{
    \rset{a} \subseteq \rset{b}
    \\
    \sesst_a \suba \varsesst_a{}^{(a\in\rset{a})}
  }{
    \bigsor_{a\in\rset{a}} a.\sesst_a\optional{\sor\cend}_{\win\in\rset{a}}
    \suba
    \bigsor_{b\in\rset{b}} b.\varsesst_b\optional{\sor\cend}_{\win\in\rset{b}}
  }
\qquad
\inferrule[\rulename{s-output}]{
  \rset{b} \subseteq \rset{a}
  \\
  \sesst_a \suba \varsesst_a{}^{(a\in\rset{b})}
}{
  \bigsand_{a\in\rset{a}} \co{a}.\sesst_a\optional{\sand\cend}_{\win\in\rset{a}}
  \suba
  \bigsand_{b\in\rset{b}} \co{b}.\varsesst_b\optional{\sand\cend}_{\win\in\rset{b}}
}
\end{array}
\end{math}
}
\end{table}

The proof of the normal form lemma is constructive and provides an
effective procedure for rewriting every session type in its normal
form using the laws in Table~\ref{tab:simpl}.
What remains to do now is to provide the subtyping algorithm for
session types in normal form.

\begin{definition}[algorithmic subtyping]
  Let $\suba$ be the least relation defined by axioms and rules in
  Table~\ref{tab:axioms}.
\end{definition}

Because of the interpretation of $\sand$ and $\sor$ as respectively
intersections and unions, the algorithm looks embarrassingly obvious
although it states well-known properties of channel types. In
particular, rule~\rulename{s-input} states that it is safe to replace
a channel $\channel$ having some input capabilities ($\rset{b}$) with
another one $\varchannel$ having fewer input capabilities ($\rset{a}
\subseteq \rset{b}$), because any process originally using $\channel$
will be ready to handle any message $b\in\rset{b}$.
Dually, rule~\rulename{s-output} states that is safe to replace a
channel $\channel$ having some output capabilities ($\rset{b}$) with
another one $\varchannel$ having greater output capabilities
($\rset{a} \supseteq \rset{b}$), since the process originally using
$\channel$ will exercise on $\varchannel$ only a subset of the
capabilities allowed on it.
Observe that~\rulename{s-output} and~\rulename{s-input} are just
specializations of the well-known laws $\sesst \sand \varsesst \suba
\sesst$ and $\sesst \suba \sesst \sor \varsesst$ concerning
intersection and union types.
Rules~\rulename{s-bottom} and~\rulename{s-top} state obvious facts
about $\cbottom$ and $\ctop$ being the smallest and the largest
session types, respectively. Observe that rule~\rulename{s-input} is
the counterpart of rule~\rulename{s-bottom} when $\rset{a} =
\emptyset$ and the larger session type is a union. Dually, the
rule~\rulename{s-output} is the counterpart of rule~\rulename{s-top}
when $\rset{b} = \emptyset$ and the smallest session type is an
intersection.
Rule~\rulename{s-end} is required for the algorithm to be complete: it
basically states the reflexivity of $\suba$ on $\cend$.

The subtyping algorithm is correct and complete with respect to the
set of session types in normal form:

\begin{theorem}
\label{thm:algorithm}
  Let $\sesst$ and $\varsesst$ be in normal form. Then $\sesst \subs
  \varsesst$ if and only if $\sesst \suba \varsesst$.
\end{theorem}

\subsection{Type Checking}
\label{sec:checker}

\begin{table}
\caption{\label{tab:checker}\strut Type checking rules.}
\framebox[\textwidth]{
\begin{math}
\displaystyle
\begin{array}{c}
\inferrule[\rulename{t-nil}\hspace{-1em}]{\ }{
  \cbottom \vdash \pnull
}
\qquad
\inferrule[\rulename{t-end}]{\ }{
  \cend \vdash \pwin
}
\qquad
\inferrule[\rulename{t-send}]{
  \sesst \vdash \process
}{
  \co a.\sesst \vdash \co a.\process
}
\qquad
\inferrule[\rulename{t-receive}]{
  \sesst_{a_i} \vdash \process_i~{}^{(i\in I)}
}{
  \bigsor_{i\in I} a_i.\sesst_{a_i} \vdash \sum_{i\in I} a_i.\process_i
}
\qquad
\inferrule[\rulename{t-choice}]{
  \sesst \vdash \process
  \\
  \sesst \vdash \varprocess
}{
  \sesst \vdash \process \oplus \varprocess
}
\qquad
\inferrule[\rulename{t-sub}]{
  \sesst \vdash \process
  \\
  \varsesst \subs \sesst
}{
  \varsesst \vdash \process
}
\end{array}
\end{math}
}
\end{table}

We conclude with the definition of a type checker to derive judgments
of the form $\sesst \vdash \process$ meaning that $\process$ is a
well-typed process using a channel with type $\sesst$.  The type
checker is defined by the axioms and rules in Table~\ref{tab:checker}.
We abbreviate $a_1.\process_1 + \cdots + a_n.\process_n$ with
$\sum_{i\in\{1,\dots,n\}} a_i.\process_i$.

Because of the similarities between processes and session types, at
first sight the type checker looks as stating a trivial correspondence
between the two languages, but there are some lurking subtleties.
Rules~\rulename{t-nil}, \rulename{t-end}, and~\rulename{t-send} are
indeed fairly obvious: the deadlocked server $\pnull$ can only use a
channel typed by $\cbottom$ since no client can interact with it; the
terminated server $\pwin$ can use a channel typed by $\cend$ since it
has successfully ended any interaction; the server $\co{a}.\process$
sending a message $a$ can use a channel typed by $\co{a}.\sesst$ if
the continuation $\process$ uses the channel according to $\sesst$.
Rule~\rulename{t-receive} concerns servers waiting for a message from
the set $\{ a_i \mid i \in I \}$. Intuitively, these servers can use
channels typed by $\bigvee_{i\in I} a_i.\sesst_i$ where each
continuation $\process_i$ is well typed with respect to
$\sesst_i$. However, there is the possibility that two branches of the
server are guarded by the same input action. Namely, it may be the
case that $a_i = a_j$ for some $i,j\in I$ such that $i\ne j$. As we
know, this translates into the server performing an internal choice on
how to handle such a message, nondeterministically choosing between
the continuations $\process_i$ and $\process_j$. Had we typed the
server with respect to $\bigvee_{i\in I} a_i.\sesst_i$, we would be
stating that the server is capable of dealing with all the clients in
the sets $\csem{\sesst_i \sor \sesst_j}$, which is not necessarily the
case. Therefore, in order for this typing rule to be sound, we require
that the continuations $\process_i$ and $\process_j$ of different
branches guarded by the same input action $a_i = a_j$ must be typable
with respect to the same type $\sesst_{a_i} = \sesst_{a_j}$. This way,
no matter which continuation is selected, it will be well typed.
Rule~\rulename{t-choice} presents a similar problem, since the server
$\process \oplus \varprocess$ may independently reduce to either
$\process$ or $\varprocess$. Therefore, we require both choices to be
typable with respect to the same session type $\sesst$. The attentive
reader will have noticed a close relationship between this typing rule
and standard type preservation results stating that (internal)
reductions preserve the type: in this case, from the hypotheses
$\sesst \vdash \process \oplus \varprocess$ and either $\process
\oplus \varprocess \lred{} \process$ or $\process \oplus \varprocess
\lred{} \varprocess$ we easily deduce that the residual process is
still well typed with respect to $\sesst$.
The last rule~\rulename{t-sub} is a standard subsumption rule, except
that it deals with the type of the (implicit) channel used by the
process and not with the type of the process itself. It states that if
a process is well typed with respect to some session type $\sesst$,
then it is also well typed with respect to a smaller session type
$\varsesst$.  This is consistent with the intuition that it is safe to
replace a value (in this case, a channel) with another one having a
smaller type.

\begin{example}
  In the two derivations that follow, rule~\rulename{t-sub} is
  essential for rules~\rulename{t-receive} and~\rulename{t-choice} to
  be applicable.

\def\ScoreOverhang{0pt}
\begin{tabular}{@{}c@{\qquad\qquad}c@{}}
\begin{minipage}[b]{0.4\textwidth}
\begin{prooftree}
  \AxiomC{$\cend \vdash \pwin$}
  \UnaryInfC{$
    \co{a} \vdash \co{a}
  $}
  \AxiomC{$\co{a} \sand \co{b} \subs \co{b}$}
  \BinaryInfC{$
    \co{a} \sand \co{b} \vdash \co{a}
  $}
  \AxiomC{$\cend \vdash \pwin$}
  \UnaryInfC{$
    \co{b} \vdash \co{b}
  $}
  \AxiomC{$\co{a} \sand \co{b} \subs \co{b}$}
  \BinaryInfC{$
    \co{a} \sand \co{b} \vdash \co{b}
  $}
  \BinaryInfC{$
    a.(\co{a} \sand \co{b}) \vdash a.\co{a} + a.\co{b}
  $}
\end{prooftree}
\end{minipage}
&
\begin{minipage}[b]{0.4\textwidth}
\begin{prooftree}
  \AxiomC{$\cend \vdash \pwin$}
  \UnaryInfC{$
    \co{a} \vdash \co{a}
  $}
  \AxiomC{$\co{a} \sand \co{b} \subs \co{b}$}
  \BinaryInfC{$
    \co{a} \sand \co{b} \vdash \co{a}
  $}
  \AxiomC{$\cend \vdash \pwin$}
  \UnaryInfC{$
    \co{b} \vdash \co{b}
  $}
  \AxiomC{$\co{a} \sand \co{b} \subs \co{b}$}
  \BinaryInfC{$
    \co{a} \sand \co{b} \vdash \co{b}
  $}
  \BinaryInfC{$
    \co{a} \sand \co{b} \vdash \co{a} \oplus \co{b}
  $}
  \UnaryInfC{$
    a.(\co{a} \sand \co{b}) \vdash a.(\co{a} \oplus \co{b})
  $}
\end{prooftree}
\end{minipage}
\end{tabular}

The fact that the two processes $a.\co{a} + a.\co{b}$ and $a.(\co{a}
\oplus \co{b})$ are well typed with respect to the same type
$a.(\co{a} \sand \co{b})$ provides further evidence that they are
equivalent, as informally argued in Section~\ref{sec:intro}.
\eoe
\end{example}

We conclude our study with a soundness result for the type system. If
two processes are typed by dual session types, then they are
orthogonal.

\begin{theorem}
\label{thm:checker}
If $\sesst \vdash \process$ and $\co\sesst \vdash \varprocess$, then
$\process \orthogonal \varprocess$.
\end{theorem}

There is no hypothesis concerning the viability of $\sesst$, but this
is implied. The reader can easily verify that $\sesst \vdash \process$
implies $\sesst \nseqs \ctop$, coherently with the observation that no
process is able to satisfy \emph{all} processes. As a consequence the
hypotheses $\sesst \vdash \process$ and $\co\sesst \vdash \varprocess$
are enough to ensure that $\sesst$ and its dual are viable.


\section{Concluding Remarks and Future Work}
\label{sec:conclusion}

Previous formalizations of session
types~\cite{CastagnaDezaniGiachinoPadovani09,Padovani09,BarbaneraDeLiguoro10}
are based on the observation that session types are behavioral
types. As such, they are eligible for being studied by means of the
numerous and well-developed techniques for process equivalence, and
testing equivalence in particular~\cite{CCSWithoutTau,testing}. In
this view the different modalities in which actions are offered
coincide with two known behavioral operators, the internal choice
$\oplus$ and the external choice $+$. This approach, however natural
and elegant, poses a few problems mostly due to the fact that the
external choice is sometimes an internal choice in disguise: the
language of session types may be difficult to understand to the
programmer; the resulting subtyping relation is not a pre-congruence
and is thus more difficult to use in practice; also, there are
contexts where the computation of the greatest lower bound and of the
least upper bound of session types arises naturally and these must be
computed by means of meta-operators on session types~\cite{Mezzina08}.

In this work we propose an alternative language of session types which
is not immediately related to some known process algebra. The basic
idea is that the two choices can be naturally modeled by means of
intersection and union types: the session type $\sesst \sand
\varsesst$ describes a channel having \emph{both} type $\sesst$
\emph{and} type $\varsesst$ and for this reason a process can freely
use that channel as having either type; the session type $\sesst \sor
\varsesst$ describes a channel having \emph{either} type $\sesst$
\emph{or} type $\varsesst$, therefore a process using that channel
cannot make any assumption on it unless the exchanged messages provide
enough information to disambiguate its type. The intersection and
union operators are intuitive alternatives to internal and external
choices, they provide a native mechanism to the computation of
greatest lower bounds and least upper bounds, and the subtyping
relation of the resulting theory turns out to be a pre-congruence.

It is worth noting that, in our theory, the semantics of session types
solely depends on the process language, in particular on the adopted
communication model and on the orthogonality relation. Any other
concept or result is derived by these two. In this work we have
adopted a partially asynchronous communication model, where output
messages must be consumed before the sender can engage into any other
activity, and a symmetric orthogonality relation where both processes
involved in a communication must terminate successfully if the
interaction reaches a stable state. These choices led us to rediscover
a familiar theory of session types~\cite{GayHole05} but it is
plausible to expect that different interesting theories can be
developed by varying these two seminal notions. For example, using a
truly asynchronous communication model, where an output action does
not block subsequent actions, the relation $a.\co{b} \subs \co{b}.a$
would be sound because any ``client'' of $a.\co{b}$ will eventually
receive the $b$ message that the ``server'' of $\co{b}.a$ sends ahead
of time. Using a symmetric orthogonality relation might allow us to
draw a closer comparison between our theory and more standard testing
theories~\cite{testing,CastellaniHennessy98}, where the notion of
``test'' is asymmetric.
We remark here just a few planned developments of our theory:
first of all, we want to extend the presented framework to deal with
possibly infinite session types. In principle this would amount to
using a fix point operator for determining the semantics of recursive
session types as sets of possibly infinite processes. However, the
model presented in this work may need some further technical
adjustments. To see why, consider the infinite session type determined
by the equation $\sesst = a.\sesst$ which gives rise to the semantic
equation $X = \biorth{\{ \co{a}.\process \mid \process \in X\}}$. Both
$\mbottom$ and $\mtop$ are solutions of the equation, meaning that the
semantics of a session type may not be uniquely determined. At the
same time, neither of $\mbottom$ and $\mtop$ is a satisfactory
solution because they denote non-viable session types, while we would
expect $\csem\sesst$ to contain (recursive) processes that send an
infinite number of $a$ messages.  We plan to investigate whether the
semantic model of types described in~\cite{VouillonMellies04}, which
shares many properties with ours, can be used to give a proper
semantics to infinite session types.
The second extension to the presented framework is to consider
non-atomic actions of the form $\cin\type$ and $\cout\type$ where
$\type$ is a \emph{basic type} (such as $\tint$, $\tbool$, \dots) and
actions of the form $\cin\sesst$ and $\cout\sesst$ for describing
delegations (the input and output of channels of type $\sesst$). This
will give rise to more interesting relations such as $\cout\tint \sor
\cout\treal \seqs \cout\tint$ assuming $\tint$ is a subtype of
$\treal$) and will allow us to compare more thoroughly our subtyping
relation with the existing ones~\cite{GayHole05}.
Finally, it looks like the presented approach can be easily extended
to incroporate universal and existential quantifiers in session types,
so as to model polymorphism and data encapsulation. In this way we
hope to provide semantic foundations to polymorphic session
types~\cite{Gay08}.



\paragraph*{Acknowledgments.} I am grateful to the anonymous referees
for the detailed comments and feedback on an earlier version of this
paper. I wish to thank Mariangiola Dezani, Kohei Honda, and Nobuko
Yoshida for the insightful discussions.

\bibliographystyle{plain} 
\bibliography{main}

\appendix

\section{Supplement to Section~\ref{sec:processes}}
\label{sec:extra_processes}

In this section we solely introduce some handy notation related to
processes that will be useful for the proofs in
Section~\ref{sec:extra_types}.
First we define two relations, that we dub ``may'' and ``must'',
distinguishing the fact that a process \emph{may} output some message
or is always capable to (i.e., \emph{must}) perform some input or
output action, regardless of its internal transitions.

\begin{definition}[may/must]
  Let $\mu\in\coset\nameset \cup \{ \win \}$. We say that $\process$
  \emph{may output} $\mu$, notation $\process\may\mu$, if $\process
  \wlred{\mu}$.
  Let $\mu\in\nameset\cup\coset\nameset\cup\{\win\}$. We say that
  $\process$ \emph{must} $\mu$, notation $\process\must\mu$, if
  $\process \wlred{} \process'$ implies $\process' \wlred{\mu}$.
  We say that $\process$ \emph{may converge}, notation
  $\process\mayconverge$, if $\process \wlred{} \process'$ implies
  $\process\may\mu$ for some $\mu$; we say that $\process$ \emph{must
    converge}, notation $\process\mustconverge$, if there exists $\mu$
  such that $\process\must\mu$.
\end{definition}

We will sometimes say that a process $\process$ guarantees action
$\mu$ if $\process \must \mu$.

Then, we define the continuation of a process $\process$ with respect
to an action $\mu$ as the combination of all the possible residuals of
$\process$ after $\mu$. This differs from the relation $\lred{\mu}$
which relates $\process$ with \emph{one particular} (not necessarily
unique) residual of $\process$ after $\mu$. 

\begin{definition}[continuation]
\label{def:continuation}
Let $\process \wlred{\mu}$. The \emph{continuation} of $\process$ with
respect to $\mu$ is defined as $\process(\mu) \eqdef
\bigoplus_{\process \wlred{}\lred{\mu} \varprocess} \varprocess$.
\end{definition}

For example, consider $\process = a.\process_1 + b.\process_2$. On the
one hand we have $\process \lred{a} \process_1$ and also $\process
\lred{b} \process_2$ namely, there are two possibly different
residuals of $\process$ after $a$ due to two different branches of the
external choice that are guarded by the same action. On the other
hand, the (unique) continuation of $\process$ after $a$ is $\process_1
\oplus \process_2$, which expresses the fact that both branches are
possible.


\section{Supplement to Section~\ref{sec:types}}
\label{sec:extra_types}

\subsection{Semantics}

We begin by gaining some familiarity with the orthogonal and the
bi-orthogonal operators and some of their properties, in particular we
provide alternative characterizations for $\orth{X}$ and $\biorth{X}$,
we prove that $\orth{(\cdot)}$ is anti-monotonic, and we state some
known properties regarding orthogonal set and set-theoretic operators.

\begin{proposition}
\label{prop:orth}
The following properties hold:
\begin{enumerate}
\item $\orth{X} = \bigcap_{\process\in X} \sem\process$;

\item $\biorth{X} = \bigcap_{X \subseteq \sem\process} \sem\process$;

\item $X \subseteq Y$ implies $\orth{Y} \subseteq \orth{X}$;

\item $\orth{X}$ is closed;

\item $\orth{(X \cup Y)} = \orth{X} \cap \orth{Y}$.
\end{enumerate}
\end{proposition}
\begin{proof}
We prove the items in order:
\begin{enumerate}
\item We have $\varprocess \in \orth{X}$ iff $X \subseteq
  \sem\varprocess$ iff $\process \orthogonal \varprocess$ for every
  $\process \in X$ iff $\varprocess \in \sem\process$ for every
  $\process \in X$ iff $\varprocess \in \bigcap_{\process \in X}
  \sem\process$.

\item By item~(1) we have $\biorth{X} = \bigcap_{\process \in
    \orth{X}} \sem\process = \bigcap_{X \subseteq \sem\process}
  \sem\process$.

\item By item~(1) we have $\orth{Y} = \bigcap_{\process\in Y} \sem\process
  \subseteq \bigcap_{\process \in X} \sem\process = \orth{X}$.

\item From Proposition~\ref{prop:closure}(1) we obtain $\orth{X}
  \subseteq \orth{\biorth{X}}$ by replacing $X$ with $\orth{X}$. From
  the same proposition and item~(3) we obtain $\orth{\biorth{X}}
  \subseteq \orth{X}$. We conclude $\orth{X} = \orth{\biorth{X}}$.

\item By item~(1) we have $\orth{(X \cup Y)} = \bigcap_{\process\in
    X\cup Y} \sem\process = \bigcap_{\process\in X} \sem\process \cap
  \bigcap_{\process\in Y} \sem\process = \orth{X} \cap \orth{Y}$.
\qedhere
\end{enumerate}
\end{proof}

It should be observed that item~(5) of the previous proposition can be
generalized to arbitrary unions, namely that
\[
\orth{\big(\bigcup_{i\in I}X_i\big)} = \bigcap_{i\in I} \orth{X_i}
\]
for arbitrary, possibly infinite family of sets $X_i$.  The reader may
also verify that $\sand$ and $\sor$ are indeed commutative and
associative operators. These properties will be silently used in some
of the proofs that follow.

We now present an auxiliary operator that is convenient in the
definition of the semantics of session types. We write
$\gsem\alpha{X}$ for the set of processes that guarantee an $\alpha$
action and whose continuation after $\alpha$ is a process in
$X$. Formally:
\[
\gsem\alpha{X} \eqdef
\begin{cases}
  \procset & \text{if $\orth{X} = \emptyset$} \\
  \{ \process \in \procset \mid \text{$\process \must \alpha$ and
    $\process(\alpha) \in X$} \} & \text{otherwise} \\
\end{cases}
\]

Using $\gsem{\cdot}{\cdot}$ one can equivalently define the
interpretation of $\alpha.\sesst$ as $\csem{\alpha.\sesst} =
\gsem{\co\alpha}{\csem\sesst}$. In particular, the orthogonal of
$\gsem\alpha{X}$ can be computed simply by turning $\alpha$ into the
corresponding co-action and by computing the orthogonal of $X$:

\begin{proposition}
\label{prop:gsem}
  $\orth{\gsem\alpha{X}} = \gsem{\co\alpha}{\orth{X}}$.
\end{proposition}
\begin{proof}
  We distinguish three cases:
\begin{itemize}
\item ($X = \emptyset$) Then $\orth{X} = \procset$ and we conclude
  $\orth{\gsem\alpha{X}} = \orth\emptyset = \procset =
  \gsem{\co\alpha}{\procset} = \gsem{\co\alpha}{\orth{X}}$.

\item ($\orth{X} = \emptyset$) Then $\orth{\gsem\alpha{X}} =
  \orth\procset = \emptyset = \gsem{\co\alpha}{\emptyset} =
  \gsem{\co\alpha}{\orth{X}}$.

\item ($X \ne \emptyset$ and $\orth{X} \ne \emptyset$) We have:
\[
\begin{array}{rcl@{\qquad}l}
  \varprocess \in \orth{\gsem\alpha{X}} & \iff &
  \forall \process \in \gsem\alpha{X}: \process \orthogonal \varprocess 
  & (\orth{X} \ne \emptyset) \\
  & \iff & \forall \process \in \gsem\alpha{X}: \varprocess\must\co\alpha \wedge \varprocess(\co\alpha) \orthogonal \process(\alpha) \\
  & \iff & \varprocess \must\co\alpha \wedge \forall \process \in \gsem\alpha{X}: \varprocess(\co\alpha) \orthogonal \process(\alpha) & (X \ne \emptyset) \\
  & \iff & \varprocess \must\co\alpha \wedge \varprocess(\co\alpha) \in\orth{X} \\
  & \iff & \varprocess \in \gsem{\co\alpha}{\orth{X}} \\
\end{array}
\]
namely $\orth{\gsem\alpha{X}} = \gsem{\co\alpha}{\orth{X}}$.\qedhere
\end{itemize}
\end{proof}

\begin{corollary}
  $X$ closed implies $\gsem\alpha{X}$ closed.
\end{corollary}
\begin{proof}
  By Proposition~\ref{prop:gsem} we have $\biorth{\gsem\alpha{X}} =
  \orth{\gsem{\co\alpha}{\orth{X}}} = \gsem\alpha{\biorth{X}} =
  \gsem\alpha{X}$.
\end{proof}

We now have all the information for showing that $\sem\sesst$ is a
closed set of processes, so that we can rewrite $\sem\sesst$ into
$\biorth{\sem\sesst}$ and viceversa, whenever useful (Proof of
Theorem~\ref{thm:dual}).

\begin{proposition}
\label{prop:closed}
For every $\sesst$, the set $\csem\sesst$ is closed.
\end{proposition}
\begin{proof}
  An easy induction on $\sesst$. The case when $\sesst = \cend$
  follows from the fact that $\pwin \in \csem\cend$, hence
  $\orth{\csem\cend} = \csem\cend$. The case when $\sesst = \sesst_1
  \sand \sesst_2$ is proved using Proposition~\ref{prop:orth}.
\end{proof}

\begin{theorem}[Theorem~\ref{thm:dual}]
  For every $\sesst$, $\csem{\co\sesst} = \orth{\csem\sesst}$.
\end{theorem}
\begin{proof}
  By induction on $\sesst$ and by cases on its shape:
\begin{itemize}
\item $\csem{\co\cbottom} = \csem{\ctop} = \procset = \orth\emptyset =
  \orth{\csem\cbottom}$.

\item $\csem{\co\ctop} = \csem\cbottom = \emptyset = \orth\procset =
  \orth{\csem\ctop}$.

\item $\csem{\co\cend} = \{ \process \in \procset \mid \process \must
  \win \} = \orth{\csem\cend}$.

\item $\csem{\co{\alpha.\varsesst}} = \csem{\co\alpha.\co\varsesst} =
  \gsem{\co\alpha}{\csem{\co\varsesst}} =
  \gsem{\co\alpha}{\orth{\csem\varsesst}} =
  \orth{\gsem\alpha{\csem\varsesst}} =
  \orth{\csem{\alpha.\varsesst}}$.

\item $\csem{\co{\sesst_1 \sand \sesst_2}} = \csem{\co\sesst_1 \sor
    \co\sesst_2} = \biorth{(\csem{\co\sesst_1} \cup
    \csem{\co\sesst_2})} = \biorth{(\orth{\csem{\sesst_1}} \cup
    \orth{\csem{\sesst_2}})} = \orth{(\biorth{\csem{\sesst_1}} \cap
    \biorth{\csem{\sesst_2}})} = \orth{(\csem{\sesst_1} \cap
    \csem{\sesst_2})} = \orth{\csem{\sesst_1 \sand \sesst_2}}$.

\item $\csem{\co{\sesst_1 \sor \sesst_2}} = \csem{\co\sesst_1 \sand
    \co\sesst_2} = \csem{\co\sesst_1} \cap \csem{\co\sesst_2} =
  \orth{\csem{\sesst_1}} \cap \orth{\csem{\sesst_2}} =
  \orth{(\csem{\sesst_1} \cup \csem{\sesst_2})} = \orth{\csem{\sesst_1
      \sor \sesst_2}}$.
  \qedhere
\end{itemize}
\end{proof}

\subsection{Subtyping Algorithm}

\begin{lemma}
\label{lem:process_nf}
Let $\sesst \equiv \bigsor_{a\in\rset{a}} a.\sesst_a\optional{\sor
  \cend}_{\win\in\rset{a}}$ and $\varsesst \equiv
\bigsand_{a\in\rset{a}} \co{a}.\varsesst_a\optional{\sand
  \cend}_{\win\in\rset{a}}$ for some $\rset{a} \subseteq \nameset \cup
\{ \win \}$ where $\sesst_a$ and $\varsesst_a$ are viable for every $a
\in \rset{a}$. Then the following properties hold:
\begin{enumerate}
\item $\process \in \csem\sesst{}$ if and only if
  $\process\mayconverge$ and $\{ \mu \mid \process \may \mu \}
  \subseteq \coset{\rset{a}}$ and $\process\may\co{a}$ implies
  $\process(\co{a}) \in \csem{\sesst_a}{}$;

\item $\process \in \csem\varsesst{}$ if and only if
  $\process\mustconverge$ and $\rset{a} \subseteq \{ \mu \mid \process
  \must \mu \}$ and $a \in \rset{a}$ implies $\process(a) \in
  \csem{\varsesst_a}{}$.
\end{enumerate}
\end{lemma}
\begin{proof}
We prove the two items in order:
\begin{enumerate}
\item Since $\csem\sesst$ is closed we have $\process \in \csem\sesst$
  if and only if $\process \in \biorth{\csem\sesst}$ if and only if
  $\orth{\csem\sesst} \subseteq \sem\process$.
  Now
\[
  \orth{\csem\sesst} =
  \orth{\left(\bigcup_{a\in\rset{a}} \gsem{\co{a}}{\csem{\sesst_a}}\optional{\cup \csem{\cend}}_{\win\in\rset{a}}\right)}
  =
  \bigcap_{a\in\rset{a}} \orth{\gsem{\co{a}}{\csem{\sesst_a}}}
  \optional{\cap \orth{\csem\cend}}_{\win\in\rset{a}}
  =
  \bigcap_{a\in\rset{a}} \gsem{a}{\orth{\csem{\sesst_a}}}
  \optional{\cap \csem\cend}_{\win\in\rset{a}}
\]
in particular $\sum_{a\in\rset{a}} a.\varprocess_a\optional{+
  \pwin}_{\win\in\rset{a}} \in \orth{\csem\sesst}$ for every
$\varprocess_a\in\orth{\csem{\sesst_a}}$.
We deduce $\process \mayconverge$ and $\process\may\mu$ implies
$\mu\in\coset{\rset{a}}$ and $\process\may\co{a}$ implies
$\process(\co{a}) \orthogonal \varprocess_a$. Since this holds for every
$\varprocess_a \in \orth{\csem{\sesst_a}}$ we have
$\orth{\csem{\sesst_a}} \subseteq \sem{\process(\co{a})}$, which is
equivalent to $\process(\co{a}) \in \csem{\sesst_a}$.

\item We have
\[
  \csem\varsesst =
  \bigcap_{a\in\rset{a}}
  \gsem{a}{\csem{\varsesst_a}}
  \optional{\cap \csem\cend}_{\win\in\rset{a}}
\]
from which we deduce that $\process \in \csem\varsesst$ if and only if
$\process\mustconverge$ and $\mu\in\rset{a}$ implies
$\process\must\mu$ and $a\in\rset{a}$ implies $\process(a) \in
\csem{\varsesst_a}$.
\qedhere
\end{enumerate}
\end{proof}

\begin{lemma}[Lemma~\ref{lem:simplify}]
  The laws in Table~\ref{tab:simpl} are sound.
\end{lemma}
\begin{proof}
  Laws~\rulename{e-prefix}, \rulename{e-bottom}, and~\rulename{e-top}
  are left as easy exercises for the reader.
  Regarding rule~\rulename{e-dist} we have $\csem{\alpha.\sesst \sand
    \alpha.\varsesst} = \csem{\alpha.\sesst} \cap
  \csem{\alpha.\varsesst} = \gsem{\co\alpha}{\csem\sesst} \cap
  \gsem{\co\alpha}{\csem\varsesst} = \{ \process \in \procset \mid
  \process \must \co\alpha \wedge \process(\co\alpha) \in \csem\sesst
  \} \cap \{ \process \in \procset \mid \process \must \co\alpha
  \wedge \process(\co\alpha) \in \csem\varsesst \} = \{ \process \in
  \procset \mid \process \must \co\alpha \wedge \process(\co\alpha)
  \in \csem\sesst \cap \csem\varsesst \} =
  \gsem{\co\alpha}{\csem\sesst \cap \csem\varsesst} =
  \gsem{\co\alpha}{\csem{\sesst \sand \varsesst}} =
  \csem{\alpha.(\sesst \sand \varsesst)}$.
  Regarding rule~\rulename{e-input-end}, let $\sesst \eqdef
  \bigsor_{a\in\rset{a}} a.\sesst_a$ and suppose by contradiction that
  $\process \in \csem{\sesst \sand \cend}$. Then $\process \in
  \csem\sesst$ and $\process \in \csem\cend$ which implies
  $\process\mayconverge$ and $\{ \mu \mid \process\may\mu \} \subseteq
  \coset{\rset{a}}$ and $\process \must \win$.
  Since $\process\may\co{a}$ and $\process\must\win$ are incompatible
  properties we deduce $\rset{a} = \emptyset$.
  Then $\sesst \seqs \cbottom$, which contradicts the hypothesis
  $\process \in \csem{\sesst \sand \cend}$.
  The proof that rule~\rulename{e-input-output} is sound is similar,
  except that in this case $\process \in \csem{\co{b}.\varsesst}$
  implies $\process \must b$.
  Regarding rule~\rulename{e-input-output-end}, let $\sesst \eqdef
  \bigsor_{a\in\rset{a}} a.\sesst_a \sor \cend$. We only need to prove
  $\sesst \sand \co{b}.\varsesst \subs \cend$ because $\sesst \sand
  \co{b}.\varsesst \subs \co{b}.\varsesst$ is obvious and
  $\co{b}.\varsesst \sand \cend \subs \sesst \sand \co{b}.\varsesst$
  follows immediately from the fact that $\cend \subs \sesst$ and the
  pre-congruence of $\subs$.
  Let $\process \in \csem{\sesst \sand \co{b}.\varsesst}$. By
  Lemma~\ref{lem:process_nf} we deduce $\process\must b$ and
  $\process\mayconverge$. The only action in $\coset\nameset \cup \{
  \win \}$ that may coexist with a guaranteed input action ($b$) is
  $\win$. Since $\process\mayconverge$ we have $\process\may\mu$
  implies $\mu=\win$, hence $\process\must\win$. We conclude
  $\process\in\csem\cend$.
  Regarding rule~\rulename{e-input-input} let $\sesst \eqdef
  \bigsor_{a\in\rset{a}}
  a.\sesst_a\optional{\sor\cend}_{\win\in\rset{a}}$ and $\varsesst
  \eqdef \bigsor_{b\in\rset{b}}
  b.\varsesst_b\optional{\sor\cend}_{\win\in\rset{b}}$.
  By Lemma~\ref{lem:process_nf} we have $\process \in \csem{\sesst
    \sand \varsesst}$ if and only if $\process \mayconverge$ and $\{
  \mu \mid \process \may \mu \} \subseteq \coset{\rset{a} \cap
    \rset{b}}$ and $\process\may\co{a}$ implies $\process(\co{a}) \in
  \csem{\sesst_a} \cap \csem{\varsesst_a}$ if and only if $\process
  \in \csem{\bigsor_{a\in\rset{a}\cap\rset{b}} a.(\sesst_a \sand
    \varsesst_a)\{{} \sor \cend\}_{\win\in\rset{a}\cap\rset{b}}}$.
\end{proof}

Some of the proofs that follow are defined by induction on the
\emph{depth} of session types. By ``depth'' of a session type we mean
the maximum number of nested actions in it. For example $a.\co{b}
\wedge c$ has depth 2, while $a.\co{b}.c$ has depth 3. The session
types $\cbottom$, $\ctop$, and $\cend$ all have depth 0.

\begin{lemma}[Lemma~\ref{lem:nf}]
  For every session type $\sesst$ there exists $\varsesst$ in normal
  form such that $\sesst \seqs \varsesst$.
\end{lemma}
\begin{proof}
  By induction on the depth of $\sesst$ and by cases on its shape.
\begin{itemize}
\item If $\sesst \equiv \ctop$ or $\sesst \equiv \cbottom$ or $\sesst
  \equiv \cend$, then $\sesst$ is already in normal form.

\item If $\sesst \equiv \alpha.\sesst'$, then by induction hypothesis
  there exists $\varsesst'$ in normal form such that $\sesst' \seqs
  \varsesst'$.
  We reason by cases on $\varsesst'$ for finding $\varsesst$ in normal
  form such that $\sesst \seqs \varsesst$:
\begin{itemize}
\item if $\varsesst' \equiv \ctop$, then $\sesst \seqs \alpha.\ctop
  \seqs \ctop$;
\item if $\varsesst' \equiv \cbottom$, then $\sesst \seqs
  \alpha.\cbottom \seqs \cbottom$;
\item in all the other cases we have $\sesst \seqs \alpha.\varsesst'$
  which is in normal form.
\end{itemize}

\item If $\sesst \equiv \sesst_1 \sand \sesst_2$, then by induction
  hypothesis there exist $\varsesst_1$ and $\varsesst_2$ in normal
  form such that $\sesst_1 \seqs \varsesst_1$ and $\sesst_2 \seqs
  \varsesst_2$.
  We reason by cases on $\varsesst_1$ and $\varsesst_2$ (symmetric
  cases omitted):
\begin{itemize}
\item if $\varsesst_1 \equiv \cbottom$ we have $\sesst \seqs \cbottom
  \sand \varsesst_2 \seqs \cbottom$;

\item if $\varsesst_1 \equiv \ctop$ we have $\sesst \seqs \ctop \sand
  \varsesst_2 \seqs \varsesst_2$;

\item if $\varsesst_1 \equiv \bigsor_{a\in\rset{a}}
  a.\varsesst_{1,a}\optional{\sor \cend}_{\win\in\rset{a}}$ and
  $\varsesst_2 \equiv \bigsor_{b\in\rset{b}}
  b.\varsesst_{2,b}\optional{\sor \cend}_{\win\in\rset{b}}$, then by
  rule~\rulename{e-input-input} we have $\sesst \seqs \varsesst_1
  \sand \varsesst_2 \seqs \bigsor_{a\in\rset{a} \cap\rset{b}}
  a.(\varsesst_{1,a} \sand \varsesst_{2,a})\optional{\sor
    \cend}_{\win\in\rset{a}\cap\rset{b}}$.
  By induction hypothesis there exists $\varsesst_a$ in normal form
  such that $\varsesst_{1,a} \sand \varsesst_{2,a} \seqs \varsesst_a$
  for every $a \in A \cap B$, therefore $\sesst \seqs
  \bigsor_{a\in\rset{a}\cap\rset{b}} a.\varsesst_a\optional{\sor
    \cend}_{\win\in\rset{a}\cap\rset{b}}$.

\item if $\varsesst_1 \equiv \bigsand_{a\in\rset{a}}
  \co{a}.\varsesst_{1,a}\optional{\sand \cend}_{\win\in\rset{a}}$ and
  $\varsesst_2 \equiv \bigsand_{b\in\rset{b}}
  \co{b}.\varsesst_{2,a}\optional{\sand \cend}_{\win\in\rset{b}}$, then
  $\sesst \seqs \varsesst_1 \sand \varsesst_2 \seqs
  \bigsand_{a\in\rset{a}\setminus\rset{b}} \co{a}.\varsesst_{1,a} \sand
  \bigsand_{b\in\rset{b}\setminus\rset{a}} \co{b}.\varsesst_{2,b} \sand
  \bigsand_{a\in\rset{a} \cap\rset{b}} \co{a}.\varsesst_a
  \optional{\sand \cend}_{\win \in\rset{a}\cup\rset{b}}$ where
  $\varsesst_a$ is in normal form and $\varsesst_{1,a} \sand
  \varsesst_{2,a} \seqs \varsesst_a$ for every $a \in\rset{a}
  \cap\rset{b}$.

\item if $\varsesst_1 \equiv \bigsor_{a\in\rset{a}} a.\varsesst_{1,a}$
  and $\varsesst_2 \equiv \bigsand_{b\in\rset{b}}
  \co{b}.\varsesst_{2,a}\optional{\sand \cend}_{\win\in\rset{b}}$,
  then by rules~\rulename{e-input-end}
  and/or~\rulename{e-input-output} we conclude $\sesst \seqs
  \cbottom$.

\item if $\varsesst_1 \equiv \bigsor_{a\in\rset{a}}
  a.\varsesst_{1,a}\sor \cend$ and $\varsesst_2 \equiv
  \bigsand_{b\in\rset{b}} \co{b}.\varsesst_{2,a}\optional{\sand
    \cend}_{\win\in\rset{b}}$, then by
  rule~\rulename{e-input-output-end} we conclude $\sesst \seqs
  \bigsand_{b\in\rset{b}} \co{b}.\varsesst_{2,a}\sand\cend$.
\end{itemize}

\item If $\sesst \equiv \sesst_1 \sor \sesst_2$, then we reason in a
  dual fashion with respect to the previous case.
\qedhere
\end{itemize}
\end{proof}

\begin{theorem}[Theorem~\ref{thm:algorithm}]
  Let $\sesst$ and $\varsesst$ be in normal form. Then $\sesst \subs
  \varsesst$ if and only if $\sesst \suba \varsesst$.
\end{theorem}
\begin{proof}
  The ``if'' part is trivial since $\suba$ axiomatizes obvious
  properties of $\subs$. Regarding the ``only if'' part, we proceed by
  induction on the depth of $\sesst$ and $\varsesst$ and by cases on
  their (normal) form. We omit dual cases:
\begin{itemize}
\item ($\sesst \equiv \cbottom$) We conclude with an application of
  either~\rulename{s-bottom} or~\rulename{s-input} according to the
  form of $\varsesst$.

\item ($\varsesst \equiv \ctop$) We conclude with an application of
  either~\rulename{s-top} or~\rulename{s-output} according to the form
  of $\sesst$.

\item ($\sesst \equiv \bigsor_{a\in\rset{a}}
  a.\sesst_a\optional{\sor\cend}_{\win\in\rset{a}}$ and $\varsesst
  \equiv \bigsor_{b\in\rset{b}}
  b.\varsesst_b\optional{\sor\cend}_{\win\in\rset{b}}$)
  From the hypothesis $\sesst \subs \varsesst$ and
  Lemma~\ref{lem:process_nf} we deduce $\rset{a} \subseteq \rset{b}$
  and $\sesst_a \subs \varsesst_a$ for every $a\in\rset{a}$.
  By induction hypothesis we derive $\sesst_a \suba \varsesst_a$ for
  every $a\in\rset{a}$, and we conclude with an application of
  rule~\rulename{s-input}.

\item ($\sesst \equiv \bigsor_{a\in\rset{a}}
  a.\sesst_a\optional{\sor\cend}_{\win\in\rset{a}}$ and $\varsesst
  \equiv \bigsand_{b\in\rset{b}}
  \co{b}.\varsesst_b\optional{\sand\cend}_{\win\in\rset{b}}$)
  For every $\process \in \csem\sesst$ we have $\{ \mu \mid \process
  \may \mu \} \subseteq \coset{\rset{a}}$ and $\emptyset \ne \rset{b}
  \subseteq \{ \mu \mid \process \must \mu \}$ from which we deduce
  $\rset{a} = \rset{b} = \{ \win \}$.
  We conclude with an application of rule~\rulename{s-end}.

\item ($\sesst \equiv \bigsand_{a\in\rset{a}}
  \co{a}.\sesst_a\optional{\sand\cend}_{\win\in\rset{a}}$ and
  $\varsesst \equiv \bigsand_{b\in\rset{b}}
  \co{b}.\varsesst_b\optional{\sand\cend}_{\win\in\rset{b}}$)
  For every $\process \in \csem\sesst$ we have $\rset{a} \subseteq \{
  \mu \mid \process \must \mu \}$ implies $\rset{b} \subseteq \{ \mu
  \mid \process \must \mu \}$ meaning $\rset{b} \subseteq \rset{a}$.
  Furthermore, $\sesst_b \subs \varsesst_b$ for every $b \in \rset{b}$.
  By induction hypothesis we deduce $\sesst_b \suba \varsesst_b$ for
  every $b \in \rset{b}$, and we conclude with an application of
  rule~\rulename{s-output}.

\item ($\sesst \equiv \bigsand_{a\in\rset{a}}
  \co{a}.\sesst_a\optional{\sand\cend}_{\win\in\rset{a}}$ and
  $\varsesst \equiv \bigsor_{b\in\rset{b}}
  b.\varsesst_b\optional{\sor\cend}_{\win\in\rset{b}}$)
  For every $\process \in \csem{\sesst}$ we have $\rset{a} \subseteq
  \{ \mu \mid \process \must \mu \}$ implies $\{ \mu \mid \process
  \may \mu \} \subseteq \coset{\rset{b}}$, from which we deduce $\win
  \in \rset{a} \cap \rset{b}$.
  We conclude with an application of rule~\rulename{s-end}.
  \qedhere
\end{itemize}
\end{proof}

\subsection{Type Checker}

\begin{theorem}[Theorem~\ref{thm:checker}]
  If $\sesst \vdash \process$ and $\co\sesst \vdash \varprocess$, then
  $\process \orthogonal \varprocess$.
\end{theorem}
\begin{proof}
  It is sufficient to show that $\sesst \vdash \process$ implies
  $\process \in \orth{\csem\sesst}$ for some generic $\process$ and
  $\sesst$. Then, by Theorem~\ref{thm:dual} we have $\varprocess \in
  \orth{\csem{\co\sesst}} = \biorth{\csem\sesst} = \csem\sesst$ and we
  conclude $\process \orthogonal \varprocess$ by definition of orthogonal
  set.
  We prove that $\sesst \vdash \process$ implies $\process \in
  \orth{\csem\sesst}$ by induction on the derivation of $\sesst \vdash
  \process$ and by cases on the last rule applied:
\begin{itemize}
\item \rulename{t-nil} Then $\process = \pnull$ and $\sesst =
  \cbottom$ and we conclude $\pnull \in \procset = \orth\emptyset =
  \orth{\csem{\cbottom}}$.

\item \rulename{t-end} Then $\process = \pwin$ and $\sesst = \cend$
  and we conclude $\pwin \in \orth{\csem\cend} = \csem\cend = \{
  \process \in \procset \mid \process \must \win \}$.

\item \rulename{t-send} Then $\process = \co a.\varprocess$ and
  $\sesst = \co a.\varsesst$ for some $\varprocess$ and $\varsesst$
  such that $\varsesst \vdash \varprocess$.
  By induction hypothesis we deduce $\varprocess \in
  \orth{\csem\varsesst}$.
  We conclude $\process \in \orth{\csem\sesst} =
  \orth{\gsem{a}{\csem\varsesst}} = \gsem{\co
    a}{\orth{\csem\varsesst}}$ since $\process \must \co a$ and
  $\process(\co a) = \varprocess \in \orth{\csem\varsesst}$.

\item \rulename{t-receive} Then $\process = \sum_{i\in I}
  a_i.\process_i$ and $\sesst = \bigsor_{i\in I} a_i.\sesst_{a_i}$
  where $\sesst_{a_i} \vdash \process_i$ for every $i\in I$.
  By induction hypothesis we have $\process_i \in
  \orth{\csem{\sesst_{a_i}}}$ for every $i \in I$.
  We conclude $\process \in \orth{\sesst} = \orth{(\bigsor_{i\in I}
    a_i.\sesst_{a_i})} = \orth{\biorth{(\bigcup_{i\in I} \gsem{\co
        a_i}{\csem{\sesst_{a_i}}})}} = \orth{(\bigcup_{i\in I}
    \gsem{\co a_i}{\csem{\sesst_{a_i}}})} = \bigcap_{i\in I}
  \orth{\gsem{\co a_i}{\csem{\sesst_{a_i}}}} = \bigcap_{i\in I}
  \gsem{a_i}{\orth{\csem{\sesst_{a_i}}}}$ because $\process \must a_i$
  and $\process(a_i) = \bigoplus_{a_i = a_j} \process_j \in
  \orth{\csem{\sesst_{a_i}}}$ for every $i \in I$.

\item \rulename{t-choice} Then $\process = \process_1 \oplus
  \process_2$ where $\sesst \vdash \process_i$ for $i\in\{1,2\}$.
  By induction hypothesis we deduce $\process_i \in
  \orth{\csem\sesst}$ for $i \in \{1,2\}$, hence we conclude $\process
  \in \orth{\csem\sesst}$ because $\orth{\csem\sesst}$ is closed.

\item \rulename{t-sub} Then $\varsesst \vdash \process$ for some
  $\varsesst$ such that $\sesst \subs \varsesst$.
  By induction hypothesis we have $\process \in \orth{\csem\varsesst}$
  hence we conclude $\process \in \orth{\csem\sesst}$ since $\sesst
  \subs \varsesst$ implies $\orth{\csem\varsesst} \subseteq
  \orth{\csem\sesst}$ by Proposition~\ref{prop:orth}(3).
  \qedhere
\end{itemize}
\end{proof}


\end{document}